\documentclass{article}
\usepackage{amsmath,amssymb,enumerate,theorem}

\newcommand{\email}[1]{{\textit{Email:} \texttt{#1}}}
\newcommand{\homepage}[1]{{\textit{Web:} \texttt{#1}}}
\newcommand{\tmem}[1]{{\em #1\/}}
\newcommand{\tmmathbf}[1]{\ensuremath{\boldsymbol{#1}}}
\newcommand{\tmop}[1]{\ensuremath{\operatorname{#1}}}
\newcommand{\tmtexttt}[1]{{\ttfamily{#1}}}
\newenvironment{enumeratenumeric}{\begin{enumerate}[1.] }{\end{enumerate}}
\newenvironment{itemizedot}{\begin{itemize} }{\end{itemize}}
\newenvironment{proof}{\noindent\textbf{Proof\ }}{\hspace*{\fill}$\Box$\medskip}
\newtheorem{corollary}{Corollary}
\newtheorem{definition}{Definition}
\newtheorem{lemma}{Lemma}
{\theorembodyfont{\rmfamily\small}\newtheorem{problem}{Problem}}
\newtheorem{proposition}{Proposition}
\newtheorem{theorem}{Theorem}

\begin{document}

\title{Generalized companion matrix for approximate GCD}\author{Paola
Boito \thanks{\email{paola.boito@unilim.fr};
\homepage{http://www.unilim.fr/pages\_perso/paola.boito/index\_en.html}}\\
XLIM-DMI UMR 6172 Universit\'e de Limoges - CNRS \\  \and Olivier
Ruatta\thanks{\email{olivier.ruatta@unilim.fr}}\\
XLIM-DMI UMR 6172 Universit\'e de Limoges - CNRS}\maketitle

{\small{We study a variant of the univariate approximate GCD problem, where
the coefficients of one polynomial $f (x)$are known exactly, whereas the
coefficients of the second polynomial $g (x)$may be perturbed. Our approach
relies on the properties of the matrix which describes the operator of
multiplication by $g$in the quotient ring $\mathbb{C}[x] / (f)$. In
particular, the structure of the null space of the multiplication matrix
contains all the essential information about GCD$(f, g)$. Moreover, the
multiplication matrix exhibits a displacement structure that allows us to
design a fast algorithm for approximate GCD computation with quadratic
complexity w.r.t. polynomial degrees.}}

\section{Introduction\label{intro}}

The approximate polynomial greatest common divisor (denoted as AGCD) is a
central object of symbolic-numeric computation. The main difficulty of the
problem comes from the fact that is no universal notion of AGCD. One can find
different approaches and different notions for AGCD. We will not give a review
of all the existing work on this subject, but we will recall one of the most
popular approaches to show how our work brings a different point of view on
the problem.

The main approach to the computation of an AGCD consists in considering two
univariate polynomials whose coefficients are known with uncertainty. This
uncertainty can be the result of the fact that the polynomials have floating
point coefficients coming from previous computation (and so are subject to
round-off errors). The most frequently adopted formulation is related to
semi-algebraic optimization : given $\tilde{f}$ and $\tilde{g}$ two
approximate polynomials, find two polynomials $f$ and $g$ such that $\|
\tilde{f} - f\|$ and $\| \tilde{g} - g\|$ are small (lower than a given
tolerance for instance) and such that the degree of $\gcd (f, g)$ is maximal.
That is, one looks for the most singular system close to the input $(
\tilde{f}, \tilde{g})$. An $\varepsilon$-gcd is obtained if the conditions $\|
\tilde{f} - f\|< \varepsilon$ and $\| \tilde{g} - g\|< \varepsilon$ are
satisfied. One can try to compute the tolerance on the perturbation of the
input polynomial thanks to direct computation (for instance from a jump on
singular values of particular matrix for instance). This last approach has
received a great interest following the work of Zeng using Sylvester like
matrices ({\cite{Zeng}}). \

Here, we consider a slightly different problem. One of the polynomials, say
$f$, is known exactly (it is the result of an exact model) and the second one,
say $g$, is an approximate polynomial (result of measures or previous
approximation for instance). This case occurs in applications such as model
checking (to compare results of an exact model and measures). There are many
other instances of such a problem, such as simplification of fractions when
one of the polynomial is known exactly but the other one is not.

We give a example of such a situation. When modeling an electromagnetic
filter, one might want to parametrize its behavior with respect to the
frequency. But one may need to do so even if there are singularities and to do
so one may use Pad\'e approximations of the electromagnetic signal at each
point as a function of the frequency. In some cases of interest, one can know
all the singularities and so compute an exact polynomial called
characteristic. Pad\'e approximations are computed independently for each
point by a numerical process and denominators may have a non trivial gcd with
the ``characteristic'' polynomial. The denominators are not known exactly. So,
in order to identify unwanted common factors in denominators one has to
compute approximate gcds between an exact and non exact polynomials.

This AGCD problem can also be interpreted as an optimization problem. Given
$f$ exactly and $\tilde{g}$ approximately, compute a polynomial $g$ close to
$\tilde{g}$ such that $g$ has a maximal degree gcd with $f$. Our approach
takes advantage of the asymmetry of the problem and of the structure of the
quotient algebra $\mathbb{C}[x] / (f (x))$ (more accurately, of the
displacement rank of the multiplication operator in this algebra). So, we
address the following problem :

\begin{problem}
  Let $f (x) \in \mathbb{C}[x]$ a given polynomial and $g (x)$ another
  polynomial. Find $\tilde{g} (x)$ close to $g (x)$ (in a sense that will be
  explained) such that $f (x)$ and $\tilde{g} (x)$ have a gcd of maximal
  degree. \ 
\end{problem}

This may be also an interesting approach when one has two polynomials, one
known with high confidence and another with worse accuracy. This approach may
take advantage of this asymmetry which would not be possible for classical
framework based on Sylvester or B\'ezout matrices.

In this paper, we propose an approach and an algorithm to address this
problem. The proposed algorithm is ``fast'' since the exponent of its
complexity is better than the classical linear algebra exponent in the degree
of the input polynomials.

Organisation of the paper: The second section is devoted to some basic result
on algebra needed after, the third section gives an algebraic method for gcd
based on linear algebra, the fourth section recalls the Barnett formula
allowing to compute the multiplication matrix without division, the fifth
gives the displacement rank structure of the multiplication matrix, the sixth
describes the final algorithm and experiments before finishing with
conclusions and perpectives.

\section{Euclidian structure and quotient algebra\label{quot}}

In this section, we recall basic algebraic results allowing to understand the
principle of our approach. All material in this section can be found (even in
the non reduced case and in the multivariate setting) in {\cite{R}}. \

Assume that $\mathbb{K}$ is an algebraically closed field (here we think
about $\mathbb{C}$). Let $f (x)$ and $g (x) \in \mathbb{K}[x]$ and assume
that $f (x) = f_d \ast \underset{i = 1}{\overset{d}{\prod}} (x - \zeta_i)$ and
that $\zeta_i \neq \zeta_j$ for all $i \neq j$ in $\{1, \ldots, d\}$. Let
$\mathbb{A}=\mathbb{K}[x] / (f)$ and $\pi : \mathbb{K}[x] \longrightarrow
\mathbb{A}$ be the natural projection. For $i \in \{1, \ldots, d\}$, we
define $L_i (x) = \frac{\underset{j \neq i}{\prod} (x - \zeta_j)}{\underset{j
\neq i}{\prod} (\zeta_i - \zeta_j)}$, the $i^{\text{th}}$ Lagrange polynomial
associated to $\{\zeta_1, \ldots, \zeta_d \}$. Clearly, since $\deg (L_i) <
\deg (f)$ we have $\pi (L_i) = L_i$, for all $i \in \{1, \ldots, d\}$. Let
$\mathbb{A}^{\ast} = \tmop{Hom}_{\mathbb{K}} (\mathbb{A}, \mathbb{K})$ be
the usual dual space of $\mathbb{A}$. For all $i \in \{1, \ldots, d\}$, we
define $\tmmathbf{1}_{\zeta_i} : \mathbb{A} \longrightarrow \mathbb{K}$ by
$\tmmathbf{1}_{\zeta_i} (p) = p (\zeta_i)$ for all $p \in \mathbb{A}$. The
following lemma is obvious form the definition of the polynomials $L_i$ that
for $i$ and $j \in \{1, \ldots, d\}$, we have $L_i (\zeta_j) = \left\{
\begin{array}{l}
  1 \tmop{if} i = j\\
  0 \tmop{else}
\end{array} \right.$. This implies that the set $\{L_1, \ldots, L_d \}$ is a
basis of $\mathbb{A}$. A well known fact is that the set
$\{\tmmathbf{1}_{\zeta_1}, \ldots, \tmmathbf{1}_{\zeta_d} \}$ form a basis
$\mathbb{A}^{\ast}$ dual of the basis $\{L_1, \ldots, L_d \}$ of
$\mathbb{A}$. As a corollary, we have the Lagrange interpolation formula :
Each $p \in \mathbb{A}$ can be written $p (x) = \underset{i =
1}{\overset{d}{\sum}} \tmmathbf{1}_{\zeta_i} (p) \ast L_i (x)$. A funny
consequence is that if we choose $\{L_1, \ldots, L_d \}$ as a basis of
$\mathbb{A}$, for all $g \in \mathbb{K}[x]$, the remainder $\pi (g)$ of the
euclidian division of $g$ by $f$ is given by $(g (\zeta_1), \ldots, g
(\zeta_d))$ in the basis $\{L_1, \ldots, L_d \}$, i.e. $r = \underset{i =
1}{\overset{d}{\sum}} g (\zeta_i) L_i (x)$. In other word, divide $g$ by $f$
is equivalent to evaluate $g$ at the roots of $f$.

The general philosophy of this last proposition will allows us to make a lot
of proof in a very simple way. For example, it is very easy to see the
different operation in $\mathbb{A}$ using this representation. Let $g$ and
$h$ be to elements in $\mathbb{A}$, then we have $g + h = \underset{i =
1}{\overset{d}{\sum}} (g (\zeta_i) + h (\zeta_i)) \ast L_i (x)$ and $g \ast h
= \underset{i = 0}{\overset{d}{\sum}} (g (\zeta_i) \ast h (\zeta_i)) L_i (x)$
in $\mathbb{A}$. This allows us to avoid the use of the section $\sigma$. In
fact, the Lagrange polynomials $L_1, \ldots, L_d$ reveal a deeper structure on
the algebra $\mathbb{A}$ : The polynomials $L_1, \ldots, L_d$ are the
idempotents of $\mathbb{A}$, i.e. $L_i \ast L_j = \left\{ \begin{array}{l}
  L_i \tmop{if} i = j\\
  0 \tmop{else}
\end{array} \right.$.

Thanks to this description of the quotient algebra, it is easy to derive
algorithms for both polynomial solving and gcd computation even though the
problems are of very different nature.

Remark that we have expressed everything in the monomial basis since it is the
most widely used basis to express polynomials but we could use other bases. A
particular basis is the Chebyshev basis where all results are exactly the same
since it is a graduated basis.

\section{An algebraic algorithm for gcd computation\label{exact}}

To first give an idea on how to exploit the section above in order to design
algorithm for gcd, We recall a classical method for polynomial solving (see
{\cite{C}} for instance). Proofs are given for the sake of completeness and
because very similar ideas will lead us to the AGCD computation.

\subsection{Roots via eigenvalues}

Let $f (x) = \underset{i = 0}{\overset{d}{\sum}} f_i x^i \in \mathbb{C}[x]$
be a polynomial of degree $d$. Then we consider the matrix of the
multiplication by $x$ in $\mathbb{C}[x] / (f)$. Its matrix in the monomial
basis $1, \ldots, x^{d - 1}$ is the following:
\[ \tmop{Frob} (f) = \begin{array}{cc}
     & \left(\begin{array}{ccccc}
       1 & x & x^2 & \cdots & x^{d - 1}
     \end{array}\right)\\
     \left(\begin{array}{c}
       1\\
       x\\
       x^2\\
       \vdots\\
       x^{d - 1}
     \end{array}\right) & \text{$\left(\begin{array}{ccccc}
       0 & 0 & 0 & \cdots & - \frac{f_0}{f_d}\\
       1 & 0 & 0 & \cdots & - \frac{f_1}{f_d}\\
       0 & 1 & 0 & \cdots & - \frac{f_2}{f_d}\\
       \vdots & \vdots & \vdots & \ddots & \vdots\\
       0 & 0 & 0 & \cdots & - \frac{f_{d - 1}}{f_d}
     \end{array}\right)$}
   \end{array} \]
well known as the Froebenius companion matrix associated to $f$.

\begin{proposition}
  Let $f (x) \in \mathbb{C}[x]$be polynomial of degree $d$ with $d$ distinct
  roots $\mathcal{Z}(f) =\{z_1, \ldots, z_d \}$, then the eigenvalues of
  $\tmop{Frob} (f)$ are the roots of $f (x)$, i.e. $\tmop{Spec} (\tmop{Frob}
  (f)) =\{z_1, \ldots, z_d \}$. \ 
\end{proposition}

\begin{proof}
  It follows directly from the fact that $\tmop{Frob} (f)$ is the matrix of
  the multiplication by $x$ in \ $\mathbb{C}[x] / (f)$. But here we propose
  to give a direct proof by induction. In fact, we prove by induction that the
  characteristic polynomial of $\tmop{Frob} (f)$ is $f (x)$ itself (up to a
  sign and a scalar factor $1 / f_d$), i.e.:
  \[ \left|\begin{array}{ccccc}
       - x & 1 & 0 & \cdots & - \frac{f_0}{f_d}\\
       0 & - x & 1 & \cdots & - \frac{f_1}{f_d}\\
       \vdots & \vdots & \vdots & \ddots & \vdots\\
       0 & 0 & 0 & \cdots & - x - \frac{f_{d - 1}}{f_d}
     \end{array}\right| = - f (x) . \]
  since we have:
  \[ \left|\begin{array}{ccccc}
       - x & 1 & 0 & \cdots & - \frac{f_0}{f_d}\\
       0 & - x & 1 & \cdots & - \frac{f_1}{f_d}\\
       \vdots & \vdots & \vdots & \ddots & \vdots\\
       0 & 0 & 0 & \cdots & - x - \frac{f_{d - 1}}{f_d}
     \end{array}\right| = - x \ast \left|\begin{array}{ccccc}
       - x & 1 & 0 & \cdots & - \frac{f_1}{f_d}\\
       \vdots & \vdots & \vdots & \ddots & \vdots\\
       0 & 0 & 0 & \cdots & - x - \frac{f_{d - 1}}{f_d}
     \end{array}\right| - \frac{f_0}{f_d} = - (x \ast \tilde{f} (x) +
     \frac{f_0}{f_d}) \]
  and by assumption $\tilde{f} (x) = \frac{f (x) - f_0}{f_d \ast x}$ and
  finally we have the wanted result. Remark that this proof allows to avoid
  the condition that all the roots of $f (x)$ are distinct.
\end{proof}

Then, to compute the roots of $f (x)$ one can compute the eigenvalues of is
Froebenius companion matrix. This is the object of the method proposed
(reintroduced) by Edelman and Murakami {\cite{EM}} and revisited by Fortune
{\cite{F}} and many others trying to use the displacement structure of the
companion matrix. In fact, often, the author realized that the monomial basis
of the quotient algebra is not the most suitable one and proposed to express
the matrix of the same linear application but in other basis. In the case of
the Chebyshev basis this algorithm was already known by Barnett {\cite{B}} and
Cardinal later {\cite{C}}.

In the next section, we will also take advantage of the structure of the
quotient algebra to design an algorithm for gcd computation mainly using
linear algebra (eigenvalues are used in theory and never computed).

\subsection{Structure of quotient and gcd}

Let $f (x) \tmop{and} g (x) \in \mathbb{K}[x]$ such that they are both monic.
As above, we denote $\mathbb{A}=\mathbb{K}[x] / (f)$ and $d = \deg (f)$. We
denote denote $\{\zeta_1, \ldots, \zeta_d \}$ the set of roots of $f (x)$ and
we assume that $f (x)$ is squarefree, i.e. $\zeta_i \neq \zeta_j$ if $i \neq
j$. We define $\mathcal{M}_g : \left\{ \begin{array}{l}
  \mathbb{A} \longrightarrow \mathbb{A}\\
  h \longmapsto \pi (g h)
\end{array} \right.$ where $\pi (p) \in \mathbb{A}$ denote the remainder of
$p (x) \in \mathbb{K}[x]$ by division by $f (x)$. We denote $M_g$ the matrix
of $\mathcal{M}_g$ in the monomial basis $1, x, \ldots, x^{d - 1}$ of
$\mathbb{A}$ but other bases can be used. A matrix representing the map
$\mathcal{M}_g$ is called an extended companion matrix.

\begin{proposition}
  The eigenvalues of $\mathcal{M}_g$ are $\{g (\zeta_1), \ldots, g
  (\zeta_d)\}$.
\end{proposition}

\begin{proof}
  It is a direct corollary of the proposition \ref{mainprop} since if we write
  the matrix of this linear map in the Lagrange basis associated to
  $\{\zeta_1, \ldots, \zeta_d \}$ is
  \[ \left(\begin{array}{ccc}
       g (\zeta_1) & \cdots & 0\\
       \vdots & \ddots & \vdots\\
       0 & \cdots & g (\zeta_d)
     \end{array}\right) \]
  and gives the wanted result.
\end{proof}

Trivially, we have:

\begin{corollary}
  We have $\tmop{corank} (\mathcal{M}_g) = \deg (f) - \tmop{rank}
  (\mathcal{M}_g) = \deg (\gcd (f, g))$.
\end{corollary}

The column of index $i$ of $M_g$ is the column vector of the coefficients of
$x^{i - 1} \ast g (x)$.

Let $p_1, \ldots, p_l$ be a basis of $\tmop{Ker} (M_g)$ and let $P_1 (x),
\ldots, P_l (x)$ be the corresponding polynomials. First remark that
$\tmop{Ann}_{\mathbb{A}} (g) =\{P (x) \in \mathbb{A}| P (x) \ast g (x) =
0\}$ is an ideal of $\mathbb{A}$.

\begin{lemma}
  The ideal $\tmop{Ann}_{\mathbb{A}} (g)$ is a principal ideal. 
\end{lemma}

\begin{proof}
  Let us define
  \[ s (x) = \underset{\zeta \in \mathcal{Z}(f) \setminus (\mathcal{Z}(f) \cap
     \mathcal{Z}(g))}{\prod} (x - \zeta) . \]
  For all $h \in \tmop{Ann}_{\mathbb{A}} (g)$ it is clear that
  $\mathcal{Z}(h) \supset \mathcal{Z}(f) \backslash (\mathcal{Z}(f) \cap
  \mathcal{Z}(g))$ and then $s$ divide $h$. Furthermore $s \in
  \tmop{Ann}_{\mathbb{A}} (g)$ since in the Lagrange basis
  \[ s (x) \ast g (x) = \underset{i = 1}{\overset{d}{\sum}} s (\zeta_i) \ast
     g (\zeta_i) L_i (x) = 0. \]
  This shows that $\tmop{Ann}_{\mathbb{A}} (g) = (s)$.
\end{proof}

To compute $s (x)$, we built the matrix with columns formed by $p_1, \ldots,
p_l$ and we make a triangulation operating only on the columns. This way we
obtain the polynomial of minimal degree linear combination of $P_1 (x),
\ldots, P_l (x)$ and it is easily seen that this $s (x)$ up to a
multiplicative scalar factor.

\begin{lemma}
  The first column of a column echelon form of the matrix $K_g$ built from a
  basis of $\tmop{Ker} (M_g)$ is the generator of $\tmop{Ann}_{\mathbb{A}}
  (g)$, i.e. it is the vector of the coefficients of $s (x)$ up to a scalar
  multiplication.
\end{lemma}

\begin{proof}
  Since the columns of a column echelon form of the matrix $K_g$ are linearly
  independent, they form a basis of $\tmop{Ann}_{\mathbb{A}} (g)$ as
  $\mathbb{K}$-vector space. So $s (x)$ is a linear combination of the
  polynomials associated to those columns. The polynomial associated to the
  column echelon form of $K_g$ have all different degree (because it is an
  echelon form) and so $s (x)$ is a linear combination of those polynomial.
  Because $s (x)$ as the lowest degree possible, it is a scalar multiple of
  the polynomial associated to the first column. 
\end{proof}

\begin{proposition}
  $f (x) \wedge g (x) = \frac{f (x)}{s (x)} .$
\end{proposition}

\begin{proof}
  By construction, we have $s (x) \ast g (x) = 0 \tmop{mod} f (x)$ and so $s
  (x)$ divide $f (x)$. We also have $\gcd ( \frac{f (x)}{s (x)}, g (x)) = \gcd
  (f (x), g (x))$ since the roots of $\frac{f (x)}{s (x)}$ are the root of $f
  (x)$ where $g (x)$ vanishes. Since $\deg ( \frac{f (x)}{s (x)}) = \deg (\gcd
  (f (x), g (x))$ we have the wanted result.
\end{proof}

In all this section, we did not care if the polynomials are known in monomial
or Chebyshev basis for instance. In fact, in order to have an algebraic
algorithm, we only need to be able to perform euclidian division and this is
always the case if the polynomial basis is graduated (as for monomial,
Chebyshev, most of the orthogonal bases).

\section{Bezoutian and Barnett's formula\label{barnett}}

A classical matricial formulation of resultant is given by the B\'ezout
matrix. In this part, we recall the construction of the B\'ezout matrix and a
special factorization of the multiplication matrix expressed in the monomial
basis. This factorization is called Barnett formula (see {\cite{B}}). The
Barnett's formula allows to build the classical extended companion matrix
without using euclidian division and only stable numerical computations.
Furthermore, this factorization reveals that the extended companion matrix has
a special rank structure and we will use this fact later to design a fast
algorithm to compute AGCD.

\begin{definition}
  Let $f$ and $g \in \mathbb{C}[x]$ of degree $m$ and $n$ respectively (with
  $n \geqslant m$), we denote $\Theta_{f, g} (x, y) = \frac{f (x) g (y) - f
  (y) g (x)}{x - y} = \underset{i, j}{\sum} \theta_{i, j} x^i y^j =
  \underset{j = 0}{\overset{m - 1}{\sum}} \kappa_{f, g, j} (x) y^j$. The
  B\'ezout matrix associated with $f$ and $g$ is $B_{f, g} =
  \left(\begin{array}{c}
    \theta_{j, j}
  \end{array}\right)_{i, j \in \{0, \ldots, m - 1\}}$.
\end{definition}

Remark that since $\Theta_{f, g} (x, y) = \Theta_{f, g} (y, x)$ the matrix
$B_{f, g}$ is symmetric. The polynomials $\kappa_{f, g, j} (x)$ are univariate
polynomials of degree at most $m - 1$. One particular case of interest is when
$f = 1$. In this case the B\'ezout matrix has a Hankel structure, i.e.
$\theta_{i, j} = \theta_{i - 1, j + 1}$. In this case we denote $H_{g, i} (x)
= \kappa_{1, g, i} (x)$ for $i \in \{0, \ldots, m - 1\}$ which are called the
Horner polynomials.

\begin{proposition}
  Let $i \in \{0, \ldots, m - 1\}$, the polynomial $H_{g, i} (x) = c_{1, m -
  i} + \cdots + c_{1, m} x^i$ has degree $i$ and since they have different
  degree, they form a basis of $\mathbb{C}[x] / (g)$. Furthermore,
  $\Theta_{1, g} (x, y) = \underset{i = 0}{\overset{m - 1}{\sum}} H_{g, m - i}
  (x) y^i$.
\end{proposition}

\begin{corollary}
  The matrix $B_{1, g}$ is the basis conversion from the Horner basis $H_0,
  \ldots, M_{m - 1}$ to the monomial basis $1, x, \ldots, x^{n - 1}$ of
  $\mathbb{C}[x] / (g)$.
\end{corollary}

This leads us to the following theorem, known as Barnett formula (see
{\cite{B}}):

\begin{theorem}
  Let $M_f$ be the multiplication matrix associated to $f$ in $\mathbb{C}[x]
  / (g)$ in the monomial basis, we have:
  \[ M_f = B_{f, g} B_{1, g}^{- 1} . \]
\end{theorem}

\begin{proof}
  We have $\Theta_{f, g} (x, y) = f (x) \frac{g (y) - g (x)}{x - y} + g (x)
  \frac{f (x) - f (y)}{x - y}$ and so $f (x) \frac{g (x) - g (y)}{x - y}
  \equiv \Theta_{f, g} (x, y)$ in $\mathbb{C}[x, y] / (g (x))$. So, for each
  $i \in \{0, \ldots, m - 1\}$, we have $\Theta_{f, g, i} (x) \equiv f (x)
  \Theta_{1, g, i} (x)$. This last equality means that $B_{f, g}$ is the
  matrix of the multiplication by $f (x)$ in $\mathbb{C}[x] / (g)$. The
  result follows directly from this fact.
\end{proof}

The Barnett's formula reveals the rank structure of the multiplication matrix.
Furthermore, this formula is already known if we choose Chebyshev basis
instead of monomial basis to express the polynomials and the matrices have
exactly the same nature. \

\section{Structured matrices and asymptotically fast algorithms}

In this section, we briefly recall some basics on displacement structured
matrices and related algorithms.

\subsection{Displacement structure}

Given an integer $n$ and a complex number $\vartheta$ with $| \vartheta | =
1$, define the circulant matrix
\[ Z_{n^{}}^{\vartheta} = \left(\begin{array}{ccccc}
     0 &  &  &  & \vartheta\\
     1 & 0 &  &  & \\
     & 1 & \ddots &  & \\
     &  & \ddots & \ddots & \\
     &  &  & 1 & 0
   \end{array}\right) \in \mathbb{C}^{n \times n} . \]
Next, define the {\tmem{Toeplitz-like displacement operator}} as the linear
operator
\begin{eqnarray*}
  \nabla_T : \mathbb{C}^{m \times n} \longrightarrow \mathbb{C}^{m \times n}
  &  & \\
  &  & \\
  \nabla_T (A) = Z_m^1 A - A Z_n^{\vartheta} . &  & 
\end{eqnarray*}
A matrix $A \in \mathbb{C}^{m \times n}$ is said to be {\tmem{Toeplitz-like}}
if $\nabla_T (A)$ is a small rank matrix (where ``small'' means small with
respect to the matrix size). The number $\alpha = \tmop{rank} (\nabla (A))$ is
called the {\tmem{displacement rank}} of $A$. If $A$ is Toeplitz-like, then
there exist (non-unique) {\tmem{displacement generators}} $G \in
\mathbb{C}^{m \times \alpha}$ and $H \in \mathbb{C}^{\alpha \times n}$ such
that
\[ \nabla (A) = G H. \]
Toeplitz matrices and their inverses are examples of Toeplitz-like matrices.
Another useful example is the multiplication matrix $M_f$, which has
Toeplitz-like displacement rank equal to $2$, regardless of its size.

A similar definition holds for {\tmem{Cauchy-like}} structure; here the
relevant displacement operator is
\begin{eqnarray*}
  \nabla_C : \mathbb{C}^{m \times n} \longrightarrow \mathbb{C}^{m \times n}
  &  & \\
  \nabla_C (A) =^{} D_1 A - A^{} D_2, &  & 
\end{eqnarray*}
where $D_1$ and $D_2$ are diagonal matrices of appropriate size with disjoint
spectra. See {\cite{KS}} for a detailed description of displacement structure.

\subsection{Fast solution of displacement structured linear systems}

Gaussian elimination with partial pivoting (GEPP) is a well-known and reliable
algorithm that computes the solution of a linear system. Its arithmetic
complexity for an $n \times n$ matrix is asymptotically $\mathcal{O}(n^3)$.
But if the system matrix exhibits displacement structure, it is possible to
apply a variant of GEPP with complexity $\mathcal{O}(n^2)$. The main idea
consists in operating on displacement generators rather than on the whole
matrix; see {\cite{GKO}} for details.

Strictly speaking, the GKO algorithm performs GEPP (or, equivalently,
computes the PLU factorization) for Cauchy-like matrices. However, several
authors have pointed out (see {\cite{GKO}}, {\cite{Hei}}, {\cite{Pan}}) that
Toeplitz-like matrices can be stably and cheaply transformed into Cauchy-like
matrices; the same is true for displacement generators.

Consider, for instance, the case $\vartheta = 1$ and let $A$ be an $n \times
n$ Toeplitz-like matrix with generators $G$ and $H$. Denote by $D_0$ the
matrix $\tmop{diag} (1, e^{\pi i / n}, \ldots, e^{(n - 1) \pi i / n})$ and let
$F$ be the Fourier matrix of size $n \times n$. Then the matrix $F A D_0^{- 1}
F^H $is Cauchy-like, of the same displacement rank as $A$, with respect to the
displacement operator defined by $D_1 = \tmop{diag} (1, e^{2 \pi i / n},
\ldots, e^{2 \pi i (n - 1) / n})$ and $D_2 = \tmop{diag} (e^{\pi i / n}, e^{3
\pi i / n}, \ldots, e^{(2 n - 1) \pi i / n})$. Its Cauchy-like generators can
be computed as $\hat{G} = F G$ and $\widehat{H^{}}^H = F D_0 H^H$.

Generalization to the case of $m \times n$ rectangular matrices is possible.
In this case, the parameter $\vartheta$ should be chosen so that the spectra
of $D_1$ and $D_2$ are well separated (see {\cite{AR}} and {\cite{BB}}).

We also point out that the GKO algorithm can be adapted to pivoting techniques
other than partial pivoting ({\cite{Gu}}, {\cite{Ste}}). This is especially
useful in case of instability due to internal growth of generator entries. A
Matlab implementation of the GKO algorithm that takes into account several
pivoting strategies is found in the package DRSolve described in {\cite{AR}}.
In our implementation, we use the pivoting strategy proposed in {\cite{Gu}}.

\section{A structured approach to AGCD computation}

We propose here an algorithm that exploits the algebraic and displacement
structure of the multiplication matrix to compute the AGCD of two given
polynomials with real coefficients (as defined in section \ref{intro}).

\subsection{Rank estimation}

It has been pointed out in Section \ref{exact} that the rank deficiency of the
multiplication matrix equals the AGCD degree. Here we use the structured
pivoted LU decomposition to estimate the approximate rank of the
multiplication matrix. Recall that $M_g$ has a Toeplitz-like structure with
displacement rank 2; it can then be transformed into a Cauchy-like matrix
$\hat{M}_g$ as described in Section 5.2. Fast pivoted Gauss elimination yields
a factorization $\hat{M}_g = P_1 L U P_2$, where $L$ is a square, nonsingular,
lower triangular matrix with diagonal entries equal to 1, $U$ is upper
triangular and $P_{1, 2}$ are permutation matrices. Inspection of the diagonal
entries (or of the row norms) of $U$ allows to estimate the approximate rank
of $\hat{M}_g$ and, therefore, of $M_g$.

\subsection{Minimization of a quadratic functional}

Let us suppose that:
\begin{itemizedot}
  \item the polynomial $f (x) = \sum_{j = 0}^n f_j x^j$ is exactly known,
  
  \item the polynomial $g (x) = \sum_{j = 0}^m g_j x^j$ is approximately known
  and may be perturbed, so that we consider its coefficients as variables,
  
  \item the AGCD degree is known. 
\end{itemizedot}
Then we can reformulate the problem of AGCD computation as the minimization
of a quadratic functional. Indeed, recall that the cofactor $v (x)$ with
respect to $f (x)$ is defined by the ``shortest'' vector (i.e., the vector
with the maximum number of trailing zeros) that belongs to the null space of
$M_g .$ We assume $v (x)$ to be monic; we denote its degree as $k$ and we have
\[ M_g v = M_g \cdot \left(\begin{array}{c}
     v_0\\
     \vdots\\
     v_{k - 1}\\
     1\\
     0\\
     \vdots\\
     0
   \end{array}\right) = \left(\begin{array}{c}
     0\\
     \vdots\\
     \vdots\\
     \vdots\\
     \vdots\\
     \vdots\\
     0
   \end{array}\right) . \]
Also observe that the entries of $M_g$ are linear functions of the
coefficients of $g (x)$. Then the equation $M_g v = 0$ can be rewritten as
$\mathcal{F}(g, v)$=0, where the functional $\mathcal{F}$ is defined as
\begin{eqnarray*}
  \mathcal{F}: \mathbb{C}^{m + 1} \times \mathbb{C}^k \longrightarrow
  \mathbb{R}_+ &  & \\
  \mathcal{F}(g, v) = \| M_g v \|_2^2 . &  & 
\end{eqnarray*}
For a preliminary study of the problem, we have chosen to solve the equation
$\mathcal{F}(g, v)$=0 by means of Newton's method, applied so as to exploit
structure. Denote by $z = [g_0, \ldots, g_m, v_0, \ldots, v_{k - 1}]^T$ the
vector of unknowns; then each Newton step has the form
\[ z_{}^{(j + 1)} = z^{(j)} - J (g^{(j)}, v^{(j)})^{\dagger} M_{g^{(j)}}
   v^{(j)} . \]
In particular, notice that the Jacobian matrix associated with $\mathcal{F}$
is an $n \times (m + k + 1)$ Toeplitz-like matrix of displacement rank $3$.
This property allows to compute a solution of the linear system $J (g^{(j)},
v^{(j)}) y = M_{g^{(j)}} v^{(j)}$ in a fast way; therefore, the arithmetic
complexity of each iteration is quadratic w.r.t. the degree of the input
polynomials.

We propose in the future to take into consideration other optimization methods
in the quasi-Newton family, such as BFGS.

\subsection{Computation of displacement generators}

In order to perform fast factorization of the multiplication matrix $M_g$, we
need to compute Toeplitz-like displacement generators. It turns out that the
range of $\nabla (M_g)$ is spanned by the first and last column of the
displaced matrix, and the columns of indices from $2$ to $n - 1$ are multiples
of the first one. Therefore, it suffices to compute a few rows and columns of
$M_g$ in order to obtain displacement generators. This can be done in a fast
and stable way by using Barnett's formula. \ If we denote as $e_j$ the $j$th
vector of the canonical basis of $\mathbb{C}^{^n}$, then the computation of
the $j$-th column of $M_g $can be seen as

$M_g (:, j) = B (f, g) \cdot \left( B (1, f)^{- 1} e_j \right),$

{\noindent}that is, it consists in solving a triangular Hankel linear system
and computing a matrix-vector product. For row computation, recall that the
Bezoutian is a symmetric matrix; we have analogously:

$M_g (j, :) = e_{j^{}}^T \cdot B (f, g) \cdot B (1, f)^{- 1} = \left( B (1,
f)^{- 1} B (f, g) e_j^T \right)^T$,

{\noindent}so that the computation of a row of $M_g$ amounts to performing a
matrix-vector product and solving a Hankel triangular system.

A similar approach holds for computation of displacement generators of the
Jacobian matrix $J (g, v)$ associated with the functional $\mathcal{F}(g, v)$.

\subsection{Description of the algorithm} \

Input: coefficients of polynomials $f (x)$ and $g (x)$.

{\noindent}Output: a perturbed polynomial $\tilde{g} (x)$ such that $f$ and
$\tilde{g}$ have a nontrivial common factor.
\begin{enumeratenumeric}
  \item Estimate the approximate rank $k$ of $M_g$ by computing a fast pivoted
  LU decomposition of the associated Cauchy-like matrix.
  
  \item Again by using fast LU, compute a vector $v = \left[ v_0, v_1, \ldots
  ., v_{k - 1}, 1, 0, \ldots ., 0 \right]^T$ in the approximate null space of
  $M_g$.
  
  \item Apply structured Newton with initial guess $(g, v)$ and compute
  polynomials $\tilde{g}$ and $\tilde{v}$ such that $f$ and $\tilde{g}$ have a
  common factor of degree $\deg f - k$ and $\tilde{v}$ is the monic cofactor
  for $f$.
\end{enumeratenumeric}

\subsection{Numerical experiments and computational issues}

We have written a preliminary implementation of the proposed method in Matlab
(available at the URL
\tmtexttt{http://www.unilim.fr/pages\_perso/paola.boito/MMgcd.m}).

The results of a few numerical experiments are shown below. The polynomials
$f$ and $g$ are monic and have random coefficients uniformly distributed over
$\left[ - 1, 1] \right.$. They have an exact GCD of prescribed degree. A
perturbation is then added to $g$. The perturbation vector has random entries
uniformly distributed over $[ - \eta, \eta]$ and its norm is of the order of
magnitude of $\eta$. We show:
\begin{itemizedot}
  \item the residual $\mathcal{F}( \tilde{g}, \tilde{v}),$
  
  \item the 2-norm distance between the exact and the computed cofactor $v$,
  
  \item the 2-norm distance between the exact and the computed perturbed
  polynomial $g$ (which is expected to be roughly of the same order of
  magnitude as $\eta$).
\end{itemizedot}
In the following table we have taken $\eta =$1e-5.

\begin{tabular}{|l|l|l|l|}
  \hline
  $n, m, \deg f$ & $\mathcal{F}( \tilde{g}, \tilde{v})$ & $\| v - \tilde{v}
  \|_2$ & $\left. \| g - \tilde{g} \right\|_2$\\
  \hline
  8, 7, 3 & 1.02e-15 & 1.19e-15 & 1.40e-5\\
  \hline
  15, 14, 5 & 1.51e-15 & 2.26e-15 & 1.35e-4\\
  \hline
  22, 22, 7 & 2.07e-13 & 1.40e-13 & 2.20e-4\\
  \hline
  36, 36, 11 & 1.19e-12 & 5.07e-14 & 0.0012\\
  \hline
\end{tabular}

Here are results for $\eta =$1e-8:

\begin{tabular}{|c|c|c|c|}
  \hline
  $n, m, \deg f$ & $\mathcal{F}( \tilde{g}, \tilde{v})$ & $\| v - \tilde{v}
  \|_2$ & $\| g - \tilde{g} \|_2$\\
  \hline
  8, 7, 3 & 5.49e-15 & 1.63e-15 & 5.85e-8\\
  \hline
  28, 27, 13 & 7.90e-14 & 8.98e-14 & 6.50e-7\\
  \hline
  38, 37, 13 & 4.88e-12 & 4.26e-12 & 2.30e-5\\
  \hline
  58, 57, 23 & 2.03e-12 & 4.40e-12 & 2.54e-4\\
  \hline
\end{tabular}

There are several issues in our approach that deserve further investigation.
Let us mention in particular:
\begin{itemizedot}
  \item The choice of a threshold (or a more refined technique) for estimating
  approximate rank.
  
  \item Normalization of polynomials: here we mostly work with monic
  polynomials, but other normalizations may be considered.
  
  \item The structured implementation of the optimization step (minimizing
  $\mathcal{F}(g, v)$). We have used for now a heuristic structured version of
  the Gauss-Newton algorithm. Observe that each step of classical Gauss-Newton
  applied to our problem has the form $z^{(j + 1)}_{^{}} = z_{}^{(j)} -
  y^{(j)}_{}$, where $z^{(j)}$ is the vector containing the coefficients of
  the $j$-th iterate polynomials $g^{(j)} \tmop{and} v^{(j)}, \tmop{and}
  y_{}^{(j)}$ is the least-norm solution to the underdetermined system $J
  (g^{(j)}, v^{(j)}) y^{(j)} = M_{g^{(j)}} v^{(j)}$. Computing this least-norm
  solution in a structured and fast way is a difficult point that will require
  more work. Our implementation gives a solution which is not, in general, the
  least-norm one, even though it is typically quite close. Further work will
  also include a study of other possible optimization methods that lend
  themselves well to a structured approach.

\end{itemizedot}

\section{Conclusions}

We have proposed and implemented a fast structured matrix-based approach to a
variant of the AGCD problem, namely, the problem of computing an approximate
greatest common divisor of two univariate polynomials, one of which is known
to be exact. To our knowledge, this variant has been so far neglected in the
existing literature. It may be also interesting when one polynomial is known
with high accuracy and the other is not.

Our approach is based on the structure of the multiplication matrix and on
the subsequent reformulation of the problem as the minimization of a suitably
defined functional. Our choice of the multiplication matrix $M_g$ over other
resultant matrices (e.g., Sylvester, B\'ezout...) is motivated by
\begin{itemizedot}
  \item the smaller size of $M_g$, with respect e.g. to the Sylvester matrix,
  
  \item the strong link between the null space of $M_g$ and the gcd, and in
  particular the fact that the null space of $M_g$ immediately yields a gcd
  cofactor,
  
  \item the displacement structure of $M_g$,
  
  \item the possibility of computing selected rows and columns of $M_g$ in a
  stable and cheap way, thanks to Barnett's formula.
\end{itemizedot}
This is, however, a preliminary study. Further work will include
generalizations of the proposed problem and a more thorough analysis of the
optimization part of the algorithm. Furthermore, this approach can be
generalized in several intersting way:
\begin{itemizedot}
  \item using better bases then the monomial one,
  
  \item it can be extended to some multivariate setting to compute the
  co-factor of a polynomial $g$ in $\mathbb{C}[x_1, \ldots, x_n] / (f_1,
  \ldots, f_n)$ when $f_1, \ldots, f_n$ define a complete intersection since
  Barnett formula still holds,
  
  \item to compute the AGCD of $f$ with $g_1, \ldots, g_k$ where $f$ is known
  with accuracy but $g_1, \ldots, g_k$ are inaccurate, one can take $g$ as a
  linear combination of $g_1, \ldots, g_k$ with our method and succeed with a
  high probability.
\end{itemizedot}

\end{document}